\documentclass[11pt,a4wide]{article}

\usepackage[utf8]{inputenc}
\usepackage[english]{babel}

\usepackage{xspace}

\usepackage{amsmath,amsfonts,amssymb,epsfig,color}
\usepackage{stmaryrd}
\usepackage{amssymb}
\usepackage{amsfonts}
\usepackage{amsmath}
\usepackage{hyperref}
\usepackage{amsthm}
\usepackage{amssymb,amstext,hyperref}
\usepackage{cite}
\usepackage{graphics}
\usepackage{amsmath,dsfont}
\usepackage{mathrsfs}
\usepackage{fancyhdr}
\usepackage{verbatim}
\usepackage{boxedminipage}
\usepackage{colortbl}
\usepackage{algorithm}
\usepackage[noend]{algorithmic}
\usepackage{enumerate}

\usepackage[T1]{fontenc}
\usepackage{tikz}
\usepackage{graphicx}

\newtheorem{theorem}{Theorem}
\newtheorem{lemma}{Lemma}
\newtheorem{corollary}{Corollary}

\newtheorem{definition}{Definition}

\theoremstyle{plain}
\newtheorem{proposition}[theorem]{Proposition}

\newcommand{\intv}[1]{\left [ #1 \right ]}

\usepackage{aecompl}
\usepackage{color}
\definecolor{linkcol}{rgb}{0,0,0.8}
\definecolor{citecol}{rgb}{0.65,0,0}
\definecolor{titlecol}{rgb}{0.65,0,0}

\hypersetup {bookmarksopen=true, pdftitle="On the (parameterized) complexity of recognizing well-covered graphs", pdfauthor="Alves - Dabrowski - Faria - Klein - Sau - Souza",
  pdfsubject="On the (parameterized) complexity of recognizing well-covered graphs", 
  pdfmenubar=true, 
  pdfhighlight=/O, 
  colorlinks=true, 
  pdfpagemode=None, 
  pdfpagelayout=DoublePage, 
  pdffitwindow=true, 
  linkcolor=linkcol, 
  citecolor=citecol, 
  urlcolor=linkcol 
}

\newcommand{\nd}{\mbox{\sf nd}\xspace}
\newcommand{\tw}{\mbox{\sf tw}\xspace}
\newcommand{\cw}{\mbox{\sf cw}\xspace}

\renewcommand{\P}{\mbox{\sf P}\xspace}
\newcommand{\NP}{\mbox{\sf NP}\xspace}

\newcommand{\fpt}{\mbox{\sf FPT}\xspace}

\newcommand{\coNPP}{\mbox{\sf coNP\xspace}}
\newcommand{\coNPc}{\mbox{\sf {\hspace{.16cm}coNPc\hspace{.15cm}}}}
\newcommand{\coNPbis}{\mbox{\sf {coNP}}\xspace}
\newcommand{\NPc}{\mbox{\sf NPc}\xspace}
\newcommand{\NPh}{\mbox{\sf NPh}\xspace}
\newcommand{\NPcoNPh}{\mbox{\sf (co)NPh}\xspace}
\newcommand{\XP}{\mbox{\sf XP}\xspace}

\newcommand{\W}{{\sf W}\xspace}

\newtheorem{fact}{Fact}{\bfseries}{\itshape}

\begin{document}

 \date{}

\title{\vspace{-1.1cm}On the (parameterized) complexity of recognizing well-covered $(r,\ell)$-graphs\thanks{
This work was supported by FAPERJ, CNPq, CAPES Brazilian Research Agencies, EPSRC (EP/K025090/1), the Leverhulme Trust (RPG-2016-258), and the French ANR projects DEMOGRAPH (ANR-16-CE40-0028) and ESIGMA (ANR-17-CE40-0028).\newline}}

\author{Sancrey Rodrigues Alves\thanks{FAETEC, Fundação de Apoio \`a Escola T\'ecnica do Estado do Rio de Janeiro, Brazil. \texttt{sancrey@cos.ufrj.br}.} \and
 Konrad K. Dabrowski\thanks{Department of Computer Science, Durham University, Durham, United Kingdom. \texttt{konrad.dabrowski@durham.ac.uk}.} \and
  Luerbio Faria\thanks{UERJ, DICC, Universidade do Estado do Rio de Janeiro, Brazil. \texttt{luerbio@cos.ufrj.br}.} \and
  Sulamita~Klein\thanks{UFRJ, COPPE-Sistemas, Universidade Federal do Rio de Janeiro, Brazil. \texttt{sula@cos.ufrj.br}.} \and
  Ignasi Sau\thanks{CNRS, LIRMM, Universit\'e de Montpellier, Montpellier, France. \texttt{ignasi.sau@lirmm.fr}.}~~\thanks{Departamento de Matem\'atica, Universidade Federal do Cear\'a, Fortaleza, Brazil.} \and
  U\'everton S. Souza\thanks{UFF, IC, Universidade Federal Fluminense, Niter\'oi, Brazil. \texttt{ueverton@ic.uff.br}.}}

\maketitle

\vspace{-.5cm}

\begin{abstract}
  \noindent An $(r, \ell)$-\emph{partition} of a graph $G$ is a partition of its vertex set into $r$ independent sets and $\ell$ cliques. A graph is $(r, \ell)$ if it admits an $(r, \ell)$-partition. A graph is {\it {well-covered}} if every maximal independent set is also maximum. A graph is $(r,\ell)$-{\it {well-covered}} if it is both $(r,\ell)$ and well-covered. In this paper we consider two different decision problems. In the \textsc{$(r,\ell)$-Well-Covered Graph} problem (\textsc{$(r,\ell)$wc-g} for short), we are given a graph $G$, and the question is whether $G$ is an $(r,\ell)$-well-covered graph. In the \textsc{Well-Covered $(r,\ell)$-Graph} problem (\textsc{wc-$(r,\ell)$g} for short), we are given an $(r,\ell)$-graph $G$ together with an $(r,\ell)$-partition, and the question is whether $G$ is well-covered. This generates two infinite families of problems, for any fixed non-negative integers $r$ and $\ell$, which we classify as being {\sf {P}}, {\sf {coNP}}-complete, {\sf {NP}}-complete, {\sf {NP}}-hard, or {\sf {coNP}}-hard.
Only the cases {\sc {wc-$(r,0)$g}} for $r\geq 3$ remain open.
In addition, we consider the parameterized complexity of these problems for several choices of parameters, such as the size $\alpha$ of a maximum independent set of the input graph, its neighborhood diversity, its clique-width, or the number $\ell$ of cliques in an $(r, \ell)$-partition. In particular, we show that the parameterized problem of determining  whether every maximal independent set of an input graph $G$ has cardinality equal to $k$ can be reduced to the \textsc{wc-$(0,\ell)$g} problem parameterized by $\ell$. In addition, we prove that both problems are {\sf coW[2]}-hard but can be solved in {\sf XP}-time.
\end{abstract}

\vspace{.3cm}

\noindent{\bf Keywords}: well-covered graph; $(r, \ell)$-graph; {\sf {coNP}}-completeness; {\sf FPT}-algorithm; parameterized complexity; {\sf coW[2]}-hardness.\vspace{.5cm}

\section{Introduction}
\label{sec:intro}

One of the most important combinatorial problems is {\sc {Maximum Independent Set (MIS)}}, where the objective is to find a maximum sized subset $S\subseteq V$ of pairwise non-adjacent vertices in a graph $G=(V,E)$. Maximum independent sets appear naturally in a wide range of situations, and {\sc MIS} also finds a number of ``real world'' relevant applications.

Unfortunately, the decision version of {\sc {MIS}} is an \NP-complete problem~\cite{Karp72}, and thus it cannot be solved in polynomial time unless $\P=\NP$. In spite of the fact that finding a {\sl maximum} independent set is a computationally
hard problem, a {\sl maximal} independent set of a graph can easily be found in linear time. Indeed, a naive greedy algorithm for finding maximal independent sets consists simply of selecting an arbitrary vertex $v$ to add to a set $S$, and updating the current graph by removing the closed neighborhood $N[v]$ of $v$. This algorithm always outputs a maximal independent set in linear time, but clearly not all choices lead to a maximum independent set.

Well-covered graphs were first introduced by Plummer~\cite{PLUMMER197091} in 1970.
Plummer defined that ``a graph is said to be {\it {well-covered}} if every
minimal point cover is also a minimum cover''. This is equivalent
to demanding that all maximal independent set have the same cardinality. Therefore, well-covered graphs can be equivalently defined as the class of graphs for which the naive greedy algorithm discussed above {\sl always} outputs a maximum independent set.

\begin{sloppypar}
The problem of recognizing a well-covered graph, which we denote by \textsc{Well-Covered Graph}, was proved to be {\sf {coNP}}-complete by Chv\'atal and Slater~\cite{chvatal1993note} and independently by Sankaranarayana and Stewart~\cite{SaSt92}. On the other hand, the \textsc{Well-Covered Graph} problem is in $\P$ when the input is known to be a perfect graph of bounded clique size~\cite{DeanZ94} or a claw-free graph~\cite{lesk1984equimatchable,TankusT96}.
\end{sloppypar}

Let $r, \ell \geq 0$ be two fixed  integers. An $(r, \ell)$-\emph{partition} of a graph $G=(V,E)$ is a partition of $V$ into
$r$ independent sets $S^1,\ldots, $ $S^r$ and $\ell$ cliques $K^1,\ldots, $ $K^{\ell}$. For convenience, we allow these sets to be empty.
A graph is $(r, \ell)$ if it admits an $(r, \ell)$-partition. Note that the notion of $(r,\ell)$-graphs is a generalization of that of $r$-colorable graphs.

A $\P$ versus $\NP$-complete dichotomy for recognizing $(r, \ell)$-graphs was proved by Brandst{\"a}dt~\cite{Brandstadt96}: the problem is in  $\P$ if $\max\{r, \ell\} \leq 2$, and $\NP$-complete otherwise.
The class of $(r,\ell)$-graphs and its subclasses have been extensively studied in the literature. For instance, list partitions of $(r,\ell)$-graphs were studied by
Feder et al.~\cite{FederHKM03}. In another paper, Feder et al.~\cite{FederHKNP05} proved that recognizing graphs that are both chordal and  $(r,\ell)$ is in $\P$.

A graph is {\it {$(r, \ell)$-well-covered}} if it
is both $(r, \ell)$ and well-covered.
In this paper we analyze the complexity of  the {\sc {$(r,\ell )$-Well-Covered Graph}} problem, which consists of
deciding whether a graph is $(r,\ell)$-well-covered.
In particular, we give a complete classification of the complexity of this problem.

Additionally, we analyze the complexity of the {\textsc{Well-Covered-$(r,\ell)$-Graph}}  problem, which consists of deciding, given an $(r,\ell )$-graph $G=(V,E)$ together with an $(r,\ell )$ partition, whether $G$ is well-covered or not.
We classify the complexity of this problem for every pair $(r,\ell )$, except for the cases when $\ell=0$ and $r \geq 3$, which we leave open.

We note that similar restrictions have been considered in the literature. For instance, Kolay et al.~\cite{KolayPRS15} recently considered the problem of removing a small number of vertices from a perfect graph so that it additionally becomes $(r, \ell)$.

To the best of our knowledge,
this is the first time in the literature
that a decision problem obtained by ``intersecting'' two recognition \NP-complete and \coNPP-complete properties has been studied.
From our results, the {\sc {$(r,\ell)$wc-g}} problem has a very peculiar property, namely that some cases of the problem are in {\sf NP}, but other cases are in {\sf coNP}. And if $\P \ne \NP$, there are some cases where the decision problem is neither in NP nor in coNP.

In addition, according to the state of the art for the {\sc {Well-Covered Graph}} problem, to the best of our knowledge this is the first work that associates the hardness of {\sc {Well-Covered Graph}} with the number of independent sets and the number of cliques of an $(r,\ell)$-partition of the input graph. This shows an important structural property for classifying the complexity of subclasses of well-covered graphs.

As a by-product of this paper, an infinite class of decision problems was classified as being both {\sf NP}-hard and {\sf coNP}-hard. Hence, unless $\P=\NP$
these decision problems are neither in {\sf NP} nor in {\sf coNP}.

More formally, in this paper we focus on the following two decision problems.

\begin{flushleft}
\fbox{
\begin{minipage}{\textwidth}
\textsc{$(r,\ell)$-Well-Covered Graph $\left( (r,\ell)\right.$wc-g$\left.\right)$}\\
\begin{tabular}{rl}
{\bf Input:}    & A graph $G$.\\
{\bf Question:} & Is $G$ $(r, \ell)$-well-covered?
\end{tabular}
\end{minipage}}
\smallskip
\end{flushleft}

\begin{flushleft}
\fbox{
\begin{minipage}{\textwidth}
\textsc{Well-Covered $(r,\ell)$-Graph $\left(\right.$wc-$(r,\ell)$g$\left.\right)$}\\
\begin{tabular}{rl}
{\bf Input:}    & An $(r,\ell)$-graph $G$, together with a partition of $V(G)$ into\\
                & $r$ independent sets and $\ell$ cliques.\\
{\bf Question:} & Is $G$ well-covered?
\end{tabular}
\end{minipage}}
\medskip
\end{flushleft}

We establish an almost complete
characterization of the complexity of the
{\sc {$(r,\ell)$wc-g}} and {\sc wc-{$(r,\ell)$g}} problems. Our results are shown in the following tables, where $r$ (resp.~$\ell$) corresponds to the rows (resp. columns) of the tables, and where $\coNPc$ stands for {\sf {coNP}}-complete, $\NPh$ stands for {\sf {NP}}-hard, $\NPc$ stands for {\sf {NP}}-complete, and $\NPcoNPh$ stands for both {\sf {NP}}-hard and {\sf {coNP}}-hard. The symbol `?' denotes that the complexity of the corresponding problem is open.

\medskip

{\small {
\[
\begin{array}{c}
\begin{array}{|c|c|c|c|c|}
\hline
{\mbox {\sc {$(r,\ell)$wc-g}}}
 & 0 & 1 & 2               & \geq 3 \\
\hline
0        & - & \P & \P               & \NPc \\
\hline
1        & \P & \P & \P  & \NPc \\
\hline
2        & \P  & \coNPc  & \coNPc  & \NPcoNPh \\
\hline
\geq 3        & \NPh & \NPcoNPh & \NPcoNPh    & \NPcoNPh \\
\hline
\end{array}
\\
\\
\\
\begin{array}{|c|c|c|c|c|}
\hline
{\mbox {\sc {wc-{$(r,\ell)$g}}}} & 0 & 1 & 2               & \geq 3 \\
\hline
0        & - & \P & \P               & \P \\
\hline
1        & \P & \P & \P  & \P \\
\hline
2        & \P & \coNPc & \coNPc  & \coNPc \\
\hline
\geq 3        & {\mbox {\hspace{.09cm}  ? \hspace{.09cm}}} & \coNPc & \coNPc & \coNPc \\
\hline
\end{array}
\end{array}
\]
}}

\medskip

We note the following simple facts, which we will use to fill the above tables:

\begin{fact}\label{fact1}
If {\sc {$(r,\ell)$wc-g}}
is in \P, then {\sc wc-{$(r,\ell)$g}}
is in \P.
\end{fact}

\begin{fact}\label{fact2}
If {\sc wc-{$(r,\ell)$g}}
is {\sf {coNP}}-hard, then {\sc {$(r,\ell)$wc-g}} is {\sf {coNP}}-hard.
\end{fact}

Note that {\sc {wc-$(r,\ell)$g}} is in {\sf {coNP}}, since a certificate for a \textsc{NO}-instance consists just of two maximal independent sets of different size. On the other hand, for {\sc {$(r,\ell)$wc-g}} we have the following facts, which are easy to verify:

\begin{fact}\label{fact_conp}
For any pair of integers $(r,\ell)$ such that the problem of recognizing an $(r,\ell)$-graph is in {\sf P}, the {\sc {$(r,\ell)$wc-g}} problem is in {\sf {coNP}}.
\end{fact}

\begin{fact}\label{fact_np}
For any pair of integers $(r,\ell)$ such that the {\sc wc-{$(r,\ell)$g}} problem is in {\sf P}, the {\sc {$(r,\ell)$wc-g}} problem is in {\sf {NP}}.
\end{fact}

In this paper we prove that $(r,\ell)${\sc {wc-g}} with $(r,\ell) \in \{ (0,1), (1,0),
(0,2), (1,1),\allowbreak  (2,0),\allowbreak (1,2)\}$ can be solved in polynomial time,
which by Fact~\ref{fact1} yields that
{{\sc {wc-}$(r,\ell)$g}} with $(r,\ell) \in \{ (0,1), (1,0),
 (0,2),(1,1), (2,0),(1,2)\}$ can also be solved in polynomial time.
On the other hand, we prove that {\sc {wc-}}$(2,1)${\sc {g}}
is {\sf {coNP}}-complete, which by Fact~\ref{fact2} and Fact~\ref{fact_conp} yields that
{\sc {$(2,1)$wc-g}} is also {\sf {coNP}}-complete.
Furthermore,  we also prove that
{\sc {wc-}}$(0,\ell)${\sc {g}}
and
{\sc {wc-}}$(1,\ell)${\sc {g}}
are  in \P, and that
$(r,\ell)${\sc {wc-g}} with $(r,\ell) \in \{(0,3), (3,0), (1,3)\}$
are {\sf {NP}}-hard.
Finally, we state and prove
a ``monotonicity'' result, namely
Theorem~\ref{mono},
stating how to extend the {\sf {NP}}-hardness or {\sf{coNP}}-hardness of {\sc{wc-$(r,\ell)$g}} (resp. {\sc {$(r,\ell)$wc-\nobreak g}}) to {\sc{wc-$(r+1,\ell)$g}} (resp. {\sc{$(r+1,\ell)
$wc-g}}), and {\sc {wc-$(r,\ell+\nobreak 1)$g}} (resp. {\sc{$(r,\ell+\nobreak 1)$wc-g}}).
Together, these results correspond to those shown in the above tables.

In addition, we consider the parameterized complexity of these problems for several choices of the parameters, such as the size $\alpha$ of a maximum independent set of the input graph, its neighborhood diversity, its clique-width or the number $\ell$ of cliques in an $(r, \ell)$-partition. We obtain several positive and negative results. In particular, we show that the parameterized problem of determining  whether every maximal independent set of an input graph $G$ has cardinality equal to $k$ can be reduced to the \textsc{wc-$(0,\ell)$g} problem parameterized by $\ell$. In addition, we prove that both problems are {\sf coW[2]}-hard, but can be solved in {\sf XP}-time.

\medskip

The rest of this paper is organized as follows. We start in Section~\ref{sec:prelim} with some basic preliminaries about graphs, parameterized complexity, and width parameters. In Section~\ref{sec:resuls} we prove our results concerning the
classical complexity of both problems, and in Section~\ref{sec:FPT-results} we focus on their parameterized complexity.
We conclude the paper with Section~\ref{sec:further}.

\section{Preliminaries}
\label{sec:prelim}

\noindent \textbf{Graphs.}  We use standard graph-theoretic notation, and we refer the reader to~\cite{Die05} for any undefined notation. A {\it {graph}} $G = (V, E)$ consists of a finite non-empty set $V$ of vertices and a set $E$ of unordered pairs (edges) of distinct elements of $V$.
If $uv \in E(G)$, then $u, v$ are said to be {\it {adjacent}}, and $u$ is said to be a {\it {neighbor}} of $v$. A \emph{clique} (resp. \emph{independent set}) is a set of pairwise adjacent (resp. non-adjacent) vertices. A \emph{vertex cover} is a set of vertices containing at least one endpoint of every edge in the graph.
The {\it {open neighborhood}} $N(v)$ or {\it {neighborhood}}, for short, of a vertex $v \in V$ is the set of vertices adjacent to $v$. The {\it {closed neighborhood}} of a vertex $v$ is defined as $N[v]=N(v)\cup \{v\}$. A \emph{dominating set} is a set of vertices $S \subseteq V$ such that $\bigcup_{v\in S}N[v] = V$.
Given $S\subseteq V$ and $v\in V$, the {\it {neighborhood $N_S(v)$ of $v$ in $S$}}
is the set $N_S(v)=N(v)\cap S$.

Throughout the paper, we let~$n$ denote the number of vertices in the input graph for the problem under consideration.

\medskip
\noindent \textbf{Parameterized complexity.} We refer the reader to~\cite{CyganFKLMPPS15,DF13,FG06,Nie06} for basic background on parameterized complexity, and we recall here only some basic definitions.
A \emph{parameterized problem} is a language $L \subseteq \Sigma^* \times \mathbb{N}$.  For an instance $I=(x,k) \in \Sigma^* \times \mathbb{N}$, $k$ is called the \emph{parameter}.
A parameterized problem is \emph{fixed-parameter tractable} ({\sf FPT}) if there exists an algorithm $\mathcal{A}$, a computable function $f$, and a constant $c$ such that given an instance $I=(x,k)$,
$\mathcal{A}$ (called an {\sf FPT}-\emph{algorithm}) correctly decides whether $I \in L$ in time bounded by $f(k) |I|^c$.

Within parameterized problems, the class {\sf W}[1] may be seen as the parameterized equivalent to the class \NP of classical optimization problems. Without entering into details (see~\cite{CyganFKLMPPS15,DF13,FG06,Nie06} for the formal definitions), a parameterized problem being {\sf W}[1]-\emph{hard} can be seen as a strong evidence that this problem is {\sl not} \fpt. The canonical example of a {\sf W}[1]-hard problem is \textsc{Independent Set} parameterized by the size of the solution\footnote{Given a graph $G$ and a parameter $k$, the problem is to decide whether there exists an independent set $S \subseteq V(G)$ such that $|S| \geq k$.}.

The class {\sf W}[2] of parameterized problems is a class that contains $\W$[1], and so the problems that are {\sf W}[2]-\emph{hard} are  even more unlikely to be \fpt than those that are {\sf W}[1]-hard (again, see~\cite{CyganFKLMPPS15,DF13,FG06,Nie06} for the formal definitions). The canonical example of a {\sf W}[2]-hard problem is \textsc{Dominating Set} parameterized by the size of the solution\footnote{Given a graph $G$ and a parameter $k$, the problem is to decide whether there exists a dominating set $S \subseteq V(G)$  such that  $|S| \leq k$.}.

For $i \in \intv{1,2}$, to transfer ${\sf W}[i]$-hardness from one problem to another, one uses an \emph{fpt-reduction}, which given an input $I=(x,k)$ of the source problem, computes in time $f(k) |I|^c$, for some computable function $f$ and a constant $c$, an equivalent instance $I'=(x',k')$ of the target problem, such that $k'$ is bounded by a function depending only on~$k$.

Hence, an equivalent definition of $\W$[1]-hard (resp. $\W$[2]-hard) problem is any problem that admits an fpt-reduction from \textsc{Independent Set} (resp. \textsc{Dominating Set}) parameterized by the size of the solution.

Even if a parameterized problem is $\W$[1]-hard or $\W$[2]-hard, it may still be solvable in polynomial time for \emph{fixed} values of the parameter; such problems are said to belong to the complexity class \XP. Formally,
a parameterized problem whose instances consist of a pair $(x,k)$ is in \XP if it can be solved by  an algorithm with
running time $f(k) |x|^{g(k)}$, where~$f,g$ are computable functions depending only on
the parameter and $|x|$ represents the input size. For example, \textsc{Independent Set} and \textsc{Dominating Set} parameterized by the solution size are easily seen to belong to \XP.

\medskip
\noindent \textbf{Width parameters.}  A \emph{tree-decomposition} of a graph $G = (V,E)$ is a pair $(T,\mathcal X)$, where $T = (I, F)$ is a tree,
and $\mathcal X = \{B_i\}, \ i\in I$ is a family of subsets of $V(G)$, called {\it bags}\index{tree
decomposition!bag|ii} and indexed by the nodes of $T$, such that
\begin{enumerate}
\setlength{\itemsep}{0pt}
\item  each vertex $v \in V$ appears in at least one bag, i.e., $\bigcup_{i\in I} B_i = V$;
\item for each edge $e = \{x, y\} \in E$, there is an $i\in I$ such that $x,y \in B_{i}$; and
\item  for each $v \in V $ the set of nodes indexed by $\{ i \mid i \in I,~v \in B_{i}
\}$ forms a subtree of~$T$.
\end{enumerate}

The {\it width} of a tree-decomposition is defined as $\max_{i \in
I} \{|B_{i}| - 1\}$. The {\it treewidth} of~$G$, denoted by $\tw(G)$, is the minimum width of a tree-decomposition of $G$.

\medskip
The {\em clique-width} of a graph~$G$, denoted by~$\cw(G)$, is defined as the minimum number of labels needed to construct~$G$, using the following four operations:
		
		\begin{enumerate}
			\item Create a single vertex~$v$ with an integer label~$\ell$ (denoted by~$\ell(v)$);
			\item Take the disjoint union (i.e., co-join) of two graphs (denoted by~$\oplus$);
			\item Join by an edge every vertex labeled~$i$ to every vertex labeled~$j$ for~$i \neq j$ (denoted by~$\eta(i,j)$);
			\item Relabel all vertices with label~$i$ by label~$j$ (denoted by~$\rho(i,j)$).
		\end{enumerate}
		An algebraic term that represents such a construction of $G$ and uses at most~$k$ labels is said to be a {\em $k$-expression} of $G$ (i.e., the clique-width of $G$ is the minimum~$k$ for which $G$ has a $k$-expression).

\begin{sloppypar}
Graph classes with bounded clique-width include cographs~\cite{Br05}, distance-hereditary graphs~\cite{G00}, graphs of bounded treewidth~\cite{CO00-2}, graphs of bounded branchwidth~\cite{ROBERTSON1991153}, and graphs of bounded rank-width~\cite{KAMINSKI20092747}.
\end{sloppypar}

\section{Classical complexity of the problems}
\label{sec:resuls}

We start with a monotonicity theorem that will be very helpful to fill the tables presented in Section~\ref{sec:intro}. The remainder  of this section is divided into four
subsections according to whether
{\sc {$(r,\ell)$wc-g}} and
{\sc {wc-$(r,\ell)$g}} are
polynomial or ``hard'' problems.

\begin{theorem}\label{mono} Let $ r,\ell \geq 0$ be two fixed integers. Then it holds that:
\begin{enumerate}[(i)]
\item if {\sc {wc-$(r,\ell)$g}} is {\sf {coNP}}-complete then {\sc {wc-$(r+1,\ell)$g}} and {\sc {wc-$(r,\ell+1)$g}} are {\sf {coNP}}-complete;

\item if {\sc {$(r,\ell)$wc-g}} is {\sf {NP}}-hard (resp. {\sf {coNP}}-hard) then {\sc {$(r,\ell+1)$wc-g}} is {\sf {NP}}-hard (resp. {\sf {coNP}}-hard);

\item supposing that $r \geq 1$, if {\sc {$(r,\ell)$wc-g}} is {\sf {NP}}-hard (resp. {\sf {coNP}}-hard) then {\sc {$(r+\nobreak 1,\allowbreak \ell)$wc-g}} is {\sf {NP}}-hard (resp. {\sf {coNP}}-hard).
\end{enumerate}
\end{theorem}

\begin{proof}
\emph{(i)} This follows immediately from the fact that every $(r,\ell)$-graph is also an $(r+\nobreak 1,\ell)$-graph and an $(r,\ell+1)$-graph.

\emph{(ii)} Let $G$ be an instance of {\sc {$(r,\ell)$wc-g}}.
Let $H$ be {an {\sc {$(r,\ell+1)$wc-g}
instance }} defined as
the disjoint union of $G$ and a clique $Z$ with $V(Z)=\{z_1,\ldots,z_{r+1}\}$.
Clearly~$G$ is well-covered if and only if~$H$ is well-covered.
If $G$ is an $(r,\ell)$-well-covered  graph then~$H$ is an $(r,\ell+1)$-well-covered graph.
Suppose $H$ is an $(r,\ell+1)$-well-covered graph, with a partition into $r$ independent sets $S^1,\ldots,S^r$ and $\ell+1$ cliques $K^1,\ldots,K^{\ell+1}$.
Each independent set $S^i$ can contain at most one vertex of the clique $Z$.
Therefore, there must be a vertex $z_i$ in some clique $K^j$.
Assume without loss of generality that there is a vertex of $Z$ in $K^{\ell+1}$.
Then~$K^{\ell+1}$ cannot contain any vertex outside of $V(Z)$, so
 we may assume that $K^{\ell+1}$ contains all vertices of $Z$.
Now $S^1,\ldots,S^r,K^1,\ldots,K^{\ell}$ is an $(r,\ell)$-partition of~$G$, so $G$ is an $(r,\ell)$-well-covered  graph.
Hence, $H$ is a \textsc{YES}-instance of {\sc {$(r,\ell+1)$wc-g}} if and only if $G$ is a \textsc{YES}-instance of {\sc {$(r,\ell)$wc-g}}.

\emph{(iii)} Let $G$ be an instance of {\sc {$(r,\ell)$wc-g}}.
Let $G'$ be an {\sc {$(r+1,\ell)$wc-g}}
instance obtained from $G$ by adding $\ell+1$ isolated vertices. (This guarantees that every maximal independent set in $G'$ contains at least $\ell+1$ vertices.)
Since $r\geq 1$, it follows that $G'$ is an $(r,\ell)$-graph if and only if $G$ is.
Clearly $G'$ is well-covered if and only if $G$ is.

Next, find an arbitrary maximal independent set in $G'$ and let $p$ be the number of vertices in this set. Note that $p \geq \ell+1$.
Let $H$ be the join of $G'$ and a set of $p$ independent vertices $Z=\{z_1,\ldots,z_p\}$, i.e., $N_H(z_i)=V(G')$ for all~$i$.
Every maximal independent set of $H$ is either $Z$ or a maximal independent set of~$G'$ and every maximal independent set of $G'$ is a maximal independent set of~$H$.
Therefore, $H$ is well-covered if and only if $G'$ is well-covered.
Clearly, if~$G'$ is an $(r,\ell)$-graph then $H$ is an $(r+1,\ell)$-graph.
Suppose~$H$ is an $(r+1,\ell)$-graph, with a partition into $r+1$ independent sets $S^1,\ldots,S^{r+1}$ and $\ell$ cliques $K^1,\ldots,K^{\ell}$.
Each clique set $K^i$ can contain at most one vertex of $Z$.
Therefore there must be a vertex $z_i$ in some independent set $S^j$.
Suppose that there is a vertex of $Z$ in $S^{r+1}$.
Then $S^{r+1}$ cannot contain any vertex outside of $Z$.
Without loss of generality, we may assume that $S^{r+1}$ contains all vertices of $Z$.
Now $S^1,\ldots,S^r,K^1,\ldots,K^{\ell}$ is an $(r,\ell)$-partition of $G$, so $G$ is an $(r,\ell)$-graph.
Thus~$H$ is a \textsc{YES}-instance of {\sc {$(r+1,\ell)$wc-g}} if and only if $G$ is a \textsc{YES}-instance of {\sc {$(r,\ell)$wc-g}}.
\end{proof}

\subsection{Polynomial cases for {\sc {wc-$(r,\ell)$g}}}

\begin{theorem}\label{0elland1ell}
{\sc {wc-$(0,\ell)$g}} and {\sc {wc-$(1,\ell)$g}} are in \P\ for every integer $\ell \geq 0$.
\end{theorem}
\begin{proof}
It is enough to prove that {\sc {wc-$(1,\ell)$g}} is in \P.
Let $V=(S,K^1,\allowbreak K^2,\allowbreak K^3,\ldots ,$ $K^{\ell})$ be a $(1,\ell)$-partition for $G$.
Then each maximal independent set $I$ of $G$ admits a partition
$I=(I_K,S\setminus N_S(I_K))$, where $I_K$ is
an independent set of $K^1\cup K^2\cup K^3\cup \cdots \cup $ $K^{\ell}$.

Observe that there are at most $O(n^{\ell})$ choices for an independent set $I_K$ of $K^1\cup K^2\cup K^3\cup \cdots \cup $ $K^{\ell}$,
which can be listed  in time $O(n^{\ell})$, since $\ell$ is constant
and $(K^1,K^2,$ $K^3,\ldots ,$ $K^{\ell})$ is given.
For each of them, we consider the independent set $I=I_K\cup (S\setminus N_S(I_K))$.
If~$I$ is not maximal (which may happen if a vertex in $(K^1\cup K^2\cup K^3\cup \cdots \cup $ $K^{\ell})\setminus I_K$ has no neighbors in~$I$), we discard this choice of~$I_K$.
Hence, we have a polynomial number
$O(n^{\ell})$ of maximal independent sets to check in order to decide whether
$G$ is a well-covered graph.
\end{proof}

\subsection{Polynomial cases for {\sc {$(r,\ell)$wc-g}}}

\begin{fact}\label{fa1}
The graph induced by a clique or by an independent set is well-covered.
\end{fact}

The following corollary is a simple application of Fact~\ref{fa1}.

\begin{corollary}\label{01}
$G$ is a $(0,1)$-well-covered graph if and only if $G$ is a $(0,1)$-graph. Similarly, $G$ is a $(1,0)$-well-covered graph if and only if $G$ is a $(1,0)$-graph.
\end{corollary}

The following is an easy observation.

\begin{lemma}\label{02}
{\sc {$(0,2)$wc-g}} can be solved in polynomial time.
\end{lemma}
\begin{proof}
By definition, a graph $G=(V,E)$ is a $(0,2)$-graph if and only if its vertex set can be partitioned into two cliques, and this can be tested in polynomial time.
It follows that every $(0,2)$-graph has maximum independent sets of size at most 2.
Let $G$ be a $(0,2)$-graph with $(0,2)$-partition $(K^1,K^2)$.
If $V$ is a clique, then $G$ is a $(0,1)$-well-covered graph,
and hence a $(0,2)$-well-covered graph.
If $V$ is not a clique, then $G$ is a $(0,2)$-well-covered graph
if and only if $G$ has no universal vertex.
\end{proof}

In the next three lemmas we give a characterization of $(1,1)$-well-covered
graphs in terms of their graph degree sequence. Note that $(1,1)$-graphs are better known in the literature as \emph{split graphs}.

\begin{lemma}\label{11<2}
Let $G=(V,E)$ be a $(1,1)$-well-covered graph
with $(1,1)$-partition $V=(S,K)$, where $S$ is a
independent set and $K$ is a clique.
If $x\in K$, then $|N_{S}(x)|\leq 1$.
\end{lemma}
\begin{proof}
Suppose that $G$ is a $(1,1)$-well-covered graph
with $(1,1)$-partition $V=(S,K)$, where $S$ is a
independent set and $K$ is a clique.
Let $I$ be a maximal independent set of $G$ such that
$x\in I\cap K$.
Suppose for contradiction that $|N_{S}(x)|\geq 2$,
and let $y,z\in N_{S}(x)$.
Since $y,z\in S$, $N_G(y), N_G(z)\subseteq K$.
Since $K$ is a clique, vertex $x$ is the only
vertex of~$I$ in $K$.
Hence, we have that
$N_G(y)\cap (I\setminus \{x\})=$
$N_G(z)\cap (I\setminus \{x\})=\emptyset$.
Therefore
$I'=(I\setminus \{x\})\cup \{y,z\}$
is an independent set of $G$
such that $|I'|=|I|+1$. Thus,
$I$ is a maximal independent set that is not
maximum, so $G$ is not well-covered.
Thus, $|N_{S}(x)|\leq 1$.
\end{proof}

\begin{lemma}\label{careca-cabeludo11}
A graph $G$ is a $(1,1)$-well-covered graph if and only if
it admits a $(1,1)$-partition $V=(S,K)$ such that
either for every $x\in K$, $|N_{S}(x)|=0$, or
for every $x\in K$, $|N_{S}(x)|=1$.
\end{lemma}
\begin{proof}
Let $G$ be a $(1,1)$-well-covered graph.
By Lemma~\ref{11<2} we have that,
given a vertex $x\in K$, either $|N_{S}(x)|=0$
or $|N_{S}(x)|=1$.
Suppose for contradiction that
there are two vertices $x,y\in K$
such that $|N_{S}(x)|=0$
and $|N_{S}(y)|=1$.
Let $z$ be the vertex of $S$
adjacent to $y$.
Let $I$ be a maximal independent set
containing vertex $y$. Note that
the vertex $x$ is non-adjacent to every vertex of
$I\setminus \{y\}$ since there is at most one vertex
of~$I$ in $K$. The same applies to
the vertex $z$. Hence, a larger independent
set $I'$, with size $|I'|=|I|+1$,
can be obtained from $I$ by
replacing vertex $y$ with the non-adjacent
vertices $x,z$, i.e.,
$I$ is a maximal independent set
of $G$ that is not maximum, a contradiction.
Thus, either for every $x\in K$, $|N_{S}(x)|=0$, or
for every $x\in K$, $|N_{S}(x)|=1$.

Conversely, suppose that there is a $(1,1)$-partition $V=(S,K)$ of $G$ such that
either for every $x\in K$, $|N_{S}(x)|=0$, or
for every $x\in K$, $|N_{S}(x)|=1$.
If $K=\emptyset$, then
$G$ is $(1,0)$ and then $G$ is well-covered.
Hence we assume
$K\ne\emptyset$.
If for every $x\in K$, $|N_{S}(x)|=0$,
then every maximal independent set consists of all the vertices of~$S$ and exactly one vertex $v \in K$.
If for every $x\in K$, $|N_{S}(x)|=1$,
then every maximal independent set is either $I=S$,
or $I=\{x\}\cup (S\setminus N_S(x))$ for some $x \in K$.
Since $|N_S(x)|=1$ we have
$|I|=1+|S|-1=|S|$, and hence $G$ is a $(1,1)$-well-covered graph.
\end{proof}

\begin{corollary}\label{cor11}
{\sc {$(1,1)$wc-g}} can be solved in polynomial time.
\end{corollary}
\begin{proof}
Since we can check in polynomial time whether $G$ is a $(1,1)$ graph~\cite{Brandstadt96}, and one can enumerate all $(1,1)$-partitions of a split graph in polynomial time,  we can solve the $(1,1)${\sc {wg-g}} problem in polynomial time.

\end{proof}

The next lemma shows that $(1,1)$-well-covered graphs can be recognized from their degree sequences.

\begin{lemma}\label{1-1char}
$G$ is a $(1,1)$-well-covered graph if and only if
there is a positive integer $k$ such that $G$ is
a graph with a $(1,1)$-partition $V=(S,K)$ where $|K|=k$,
such that the degree sequence of $V$
is either $(k,k,k,\ldots,k,$ $i_1,i_2,\ldots,i_s,$ $0,0,0,\ldots,0)$ with $\sum_{j=1}^s (i_j)= k$,
or $(k-1,k-1,k-1,\ldots,k-1,$ $0,0,0,\ldots,0)$, where the subsequences $k, \ldots, k$ (resp. $k-1, \ldots, k-1$) have length $k$.
\end{lemma}
\begin{proof}
Let $G$ be a $(1,1)$-well-covered graph.
Then $G$ admits a $(1,1)$-partition $V=(S,K)$ where $k:=|K|, k\geq 0$.
If $k=0$, then the degree sequence
is $(0,0,0,\ldots,0)$.
If $k\geq 1$, then
by Lemma~\ref{careca-cabeludo11}
either for every $x\in K$, $|N_{S}(x)|=0$, or
for every $x\in K$, $|N_{S}(x)|=1$.
If for every $x\in K$, $|N_{S}(x)|=0$,
then the degree sequence of $G$ is
$(k-1,k-1,k-1,\ldots,k-1,$ $0,0,0,\ldots,0)$.
If for every $x\in K$, $|N_{S}(x)|=1$,
then the degree sequence
of $G$ is $(k,k,k,\ldots,k,$ $i_1,i_2,\ldots,i_s,$ $0,0,0,\ldots,0)$,
with $\sum_{j=1}^s (i_j)= k$.

Suppose that there is a positive integer $k$ such that $G$ is
a graph with $(1,1)$-partition $V=(S,K)$ where $|K|=k$,
with degree sequence either $(k,k,k,\ldots,k,$ $i_1,i_2,\ldots,i_s,$ $0,0,0,\ldots,0)$,
or $(k-1,k-1,k-1,\ldots,k-1,$ $0,0,0,\ldots,0)$,
such that $\sum_{j=1}^s (i_j)= k$.
If the degree sequence
of $G$ is $(k,k,k,\ldots,k,$ $i_1,i_2,\ldots,i_s,$ $0,0,0,\ldots,0)$,
then the vertices of $K$ are adjacent
to $k-1$ vertices of $K$ and exactly
one of $S$, since the vertices
with degree $i_1,i_2,\ldots,i_s,$
have degree at most~$k$ and the vertices
with degree $0$ are isolated.
If the degree sequence
of $G$ is $(k-1,k-1,k-1,\ldots,k-1,$ $0,0,0,\ldots,0)$,
then the vertices of $K$ are adjacent
to $k-1$ vertices of $K$ and
none of $S$ and the
vertices with degree $0$ are isolated.
By Lemma~\ref{careca-cabeludo11}
we have that $G$ is a
well-covered graph.
\end{proof}

\medskip

\begin{sloppypar}
Ravindra~\cite{Rav77} gave the following characterization
of $(2,0)$-well-covered graphs.
\end{sloppypar}

\begin{proposition}[Ravindra~\cite{Rav77}]\label{rav}
Let $G$ be a connected graph.
$G$ is a $(2,0)$-well-covered graph
if and only if $G$ contains a
perfect matching $F$ such that for
every edge $e = uv$ in~$F$,
$G[N(u)\cup N(v)]$ is a complete bipartite graph.
\end{proposition}

We now prove that  Proposition~\ref{rav} leads to
a polynomial-time algorithm.

\begin{lemma}\label{l20}
{\sc {$(2,0)$wc-g}}
can be solved in polynomial time.
\end{lemma}
\begin{proof}
Assume that~$G$ is connected and consider the weighted
graph $(G,\omega)$ with $\omega :E(G) \to \{0,1\}$
satisfying $\omega(uv)=1$, if $G[N(u)\cup N(v)]$  is a complete bipartite graph, and 0 otherwise.
By  Proposition~\ref{rav}, $G$ is well-covered if and only if $(G,\omega)$ has a weighted perfect matching with weight at least $n/2$, and this can be decided in polynomial time~\cite{edmonds65}.
\end{proof}

\begin{lemma}\label{l12}
{\sc {$(1,2)$wc-g}} can be solved in polynomial time.
\end{lemma}
\begin{proof}
We can find a $(1,2)$-partition of a graph $G$ (if such a partition exists) in polynomial time~\cite{Brandstadt96}.
After that, we use the algorithm  for {\sc {wc-$(1,\ell)$g}} given by Theorem~\ref{0elland1ell}.
\end{proof}

Below we summarize the cases for which we have shown that {\sc {wc-}}$(r,\ell)${\sc {g}} or $(r,\ell)${\sc {wc-g}} can be solved in polynomial time.

\begin{theorem}
 $(r,\ell)${\sc {wc-g}} with $(r,\ell) \in \{(0,1), (0,2), (1,0), (1,1), (1,2),(2,0)\}$
and
{\sc wc-{$(r,\ell)$g}} with $r \in \{0,1\}$ or $(r,\ell) = (2,0)$
can be solved in polynomial time.
\end{theorem}

\begin{proof}
The first part follows from
Corollary~\ref{01}, Lemma~\ref{02}, Corollary~\ref{cor11},
Lemma~\ref{l12}, and Lemma~\ref{l20}, respectively.
The second part follows from
Theorem~\ref{0elland1ell},
and Lemma~\ref{l20} together with Fact~\ref{fact1}.
\end{proof}

\subsection{{\sf {coNP}}-complete cases for {\sc {wc-$(r,\ell)$g}}}

We note that the {\sc {well-Covered Graph}}  instance $G$
constructed in the reduction of Chv\'atal and Slater~\cite{chvatal1993note}
is a $(2,1)$-graph, directly implying that {\sc wc-{$(2,1)$g}}
is {\sf {coNP}}-complete.

\begin{figure}[t]
\epsfxsize=4.7cm \centerline{\epsfbox{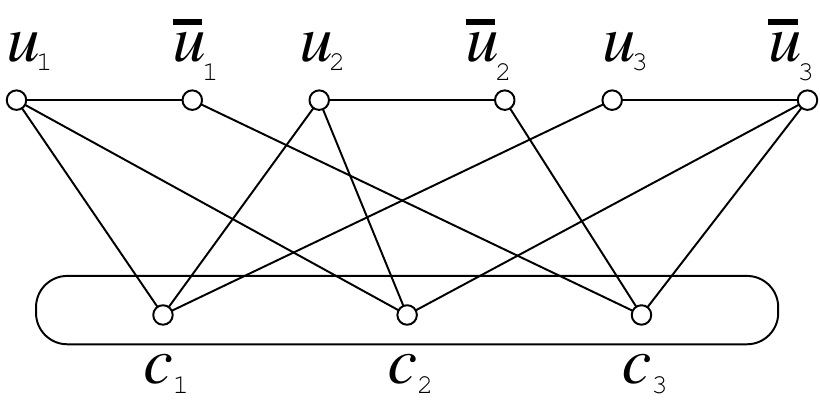}}
\caption{Chv\'atal and Slater's~\cite{chvatal1993note} {\sc {Well-Covered Graph}} instance $G=(V,E)$
obtained from the satisfiable {\sc {3-sat}} instance $I=(U,C)= \left(\{u_1,u_2,u_3\},\{(u_1,u_2,u_3),(u_1,u_2,\overline{u}_3),
(\overline{u}_1,\overline{u}_2,\overline{u}_3)\}\right)$, where $\{c_1,c_2,\ldots ,c_m\}$ is a clique of~$G$.
Observe that $I$ is satisfiable if and only if $G$ is not well-covered, since there is a maximal independent
set with size $n+1$ (e.g. $\{c_1,\overline{u}_1,\overline{u}_2,\overline{u}_3\}$) and there is a maximal independent
set of size $n$
(e.g. $\{{u}_1,{u}_2,\overline{u}_3\}$).
Note also that $G$ is a $(2,1)$-graph with $(2,1)$-partition $V=(\{u_1,u_2,\ldots ,u_n\}, \{\overline{u}_1,\overline{u}_2,\ldots ,\overline{u}_n\},
\{c_1,c_2,\ldots ,c_m\}\left.\right)$.} \label{chvatalslaterexample}
\end{figure}

\begin{sloppypar}
Indeed, Chv\'atal and Slater~\cite{chvatal1993note} take
a {\sc {3-sat}} instance $I=(U,C)=(\{u_1,u_2,\allowbreak u_3,\allowbreak \ldots ,u_n\},$
$\{c_1,c_2,c_3,\ldots ,c_m\}),$
and construct a \textsc{Well-Covered Graph} instance
$G=(V,E)=$ $\left(\right.\{u_1,u_2,u_3,$ $\ldots , u_n,\overline{u}_1,\overline{u}_2,\overline{u}_3,\ldots , \overline{u}_n,$
$c_1,c_2,c_3,\ldots ,c_m\},$ $\{xc_j:$ $x{\mbox { occurs in }} $ $c_j\}\ \cup$ $ \{u_i\overline{u}_i : 1 \leq i \leq n\}\ \cup$
$\{c_ic_j:1\leq i<j\leq m\}\left.\right)$.  Note that
$\{c_ic_j:1\leq i<j\leq m\}$ is a clique, and that
$\{u_1,u_2,u_3,$ $\ldots,  u_n\},$ and
$\{\overline{u}_1,\overline{u}_2,\overline{u}_3,\ldots , \overline{u}_n\}$
are independent sets. Hence, $G$ is
a $(2,1)$-graph. An illustration of this construction can be found in Figure~\ref{chvatalslaterexample}. This discussion can be summarized as follows.
\end{sloppypar}

\begin{proposition}[Chv\'atal and Slater~\cite{chvatal1993note}]\label{21conpc}
{\sc {wc-$(2,1)$g}} is {\sf {coNP}}-complete.
\end{proposition}

As $(2,1)$-graphs can be recognized in polynomial time~\cite{Brandstadt96}, we obtain the following.

\begin{corollary}\label{co21}
{\sc {$(2,1)$wc-g}} is  {\sf {coNP}}-complete.
\end{corollary}

\subsection{{\sf {NP}}-hard cases for {\sc {$(r,\ell)$wc-g}}}

\begin{figure}[tb]
\begin{center}
\vspace{-.35cm}
\scalebox{0.68}{
\epsfxsize=18cm \hspace*{-0.3cm}
\epsfbox{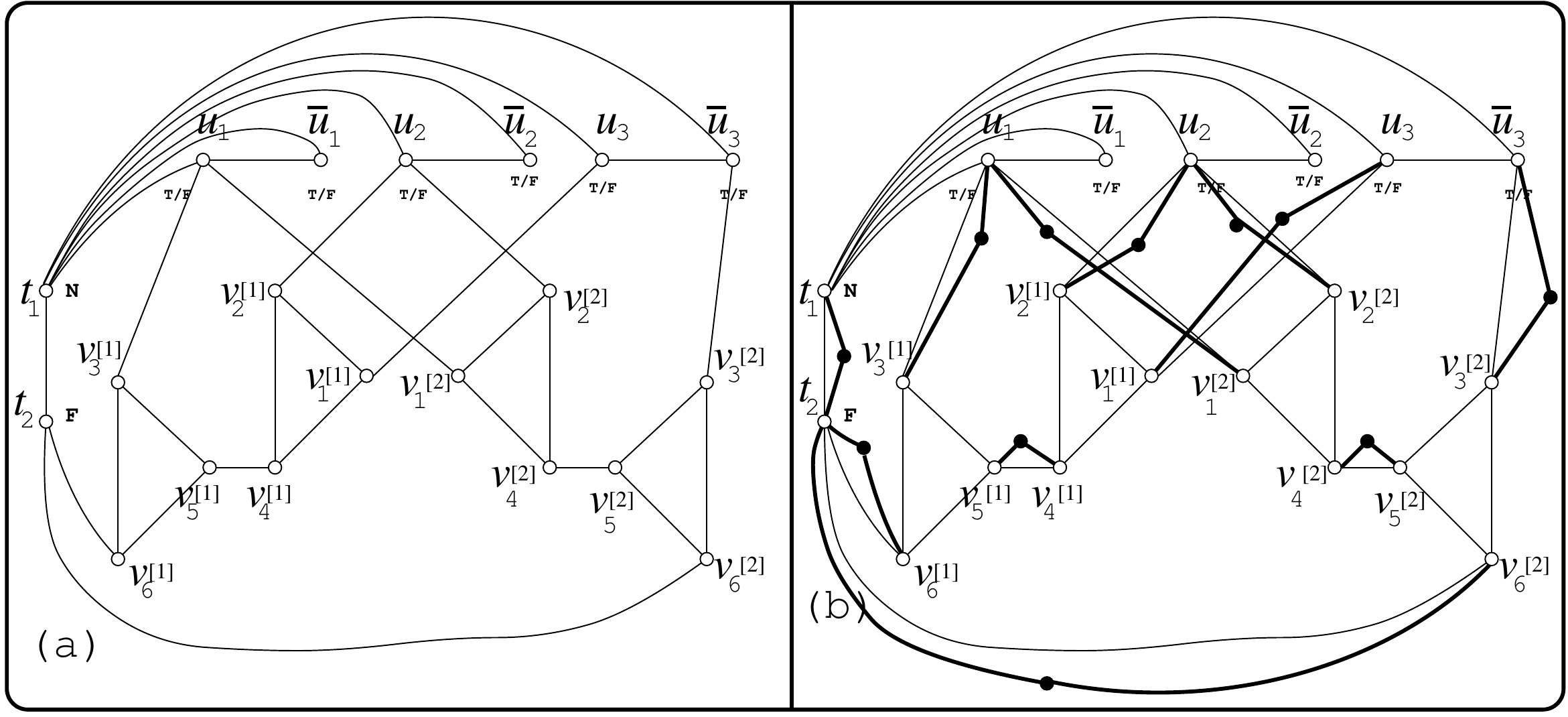}
\vspace{-2.5cm}
}
\begin{sloppypar}
\caption{(a) Stockmeyer's~\cite{Stockmeyer:1973} {\sc {3-coloring}} instance $G$
 obtained from the {\sc {3-sat}} instance $I=(U,C)= \left(\{u_1,u_2,u_3\},\{(u_3,u_2,u_1),(u_1,u_2,\overline{u}_3)\}\right)$. (b) The graph $G'$ obtained from~$G$ by
 adding a vertex $x_{uv}$ with $N_{G'}(x_{uv})=\{u,v\}$ for
every edge $uv$ of $G$ not belonging to a triangle.} \label{stockfig}
\end{sloppypar}
\end{center}
\end{figure}

Now we prove that
{\sc {$(0,3)$wc-g}} is {\sf {NP}}-complete.
For this purpose, we slightly modify an
\NP-completeness proof
of Stockmeyer~\cite{Stockmeyer:1973}.

\begin{sloppypar}
Stockmeyer's~\cite{Stockmeyer:1973}
\NP-completeness proof of 3-coloring
considers a {\sc {3-sat}} instance
$I=(U,C)= \left(\right.\{u_1,u_2,u_3,\ldots, u_n\},
\{c_1,c_2,c_3,\ldots, c_m\}\left.\right)$,
and constructs a {\sc {3-coloring}} instance  $G=(V,E)=
\left.\right(\{u_1,u_2,u_3,\ldots,u_n,\overline{u}_1,\overline{u}_2,\overline{u}_3,\ldots , \overline{u}_n\}\cup
\{v_1[j],v_2[j],v_3[j],v_4[j],v_5[j],v_6[j]: j\in \{1,2,3,\ldots,m\}\}\cup\{t_1,t_2\},\allowbreak
\{u_i\overline{u}_i: i\in\{1,2,3,\ldots ,n\}\}\cup
\{v_1[j]v_2[j],v_2[j]v_4[j],v_4[j]v_1[j],v_4[j]v_5[j],v_5[j]v_6[j],v_6[j]v_3[j],\allowbreak
v_3[j]v_5[j]: j\in\{1,2,3,\ldots , m\}\}\cup
\{v_1[j]x, v_2[j]y, v_3[j]z: c_j=(x,y,z)\}\cup
\{t_1u_i,\allowbreak t_1\bar{u}_i  : i\in\{1,2,3,\ldots, n\}\}\cup \{t_2v_6[j]:j\in \{1,2,3,\ldots , m\}\}\left.\right)$; see
 Figure~\ref{stockfig}(a).
\end{sloppypar}

\begin{lemma}\label{thm:0,3}
{\sc {$(0,3)$wc-g}} is {\sf {NP}}-complete.
\end{lemma}
\begin{proof}
As by Theorem~\ref{0elland1ell} the {\sc Well-Covered Graph} problem can be solved in polynomial time on $(0,3)$-graphs, by Fact~\ref{fact_np} {\sc $(0,3)$wc-g} is in $\NP$.

Let $I=(U,C)$ be a {\sc {3-sat}} instance.
We produce,  in polynomial time in the size of $I$, a
{\sc {$(0,3)$wc-g}} instance~$H$, such that $I$ is satisfiable if and only
if~$H$ is $(0,3)$-well-covered.
Let $G=(V,E)$ be the graph of~\cite{Stockmeyer:1973}
obtained from~$I$,
and let~$G'$ be the graph obtained from $G$
by adding to $V$ a vertex $x_{uv}$ for
every edge~$uv$ of $G$ not belonging to a triangle,
and by adding
to $E$  edges $ux_{uv}$ and~$vx_{uv}$;
see Figure~\ref{stockfig}(b).
Finally, we define $H=\overline{G'}$ as the
complement of $G'$.
Note that, by~\cite{Stockmeyer:1973},
$I$ is satisfiable if and only if $G$
is 3-colorable. Since $x_{uv}$
is adjacent to only two different colors of $G$,
clearly
$G$ is 3-colorable if and only if $G'$
is 3-colorable. Hence,
$I$ is satisfiable if and only if $H$
is a $(0,3)$-graph.
We prove next that
$I$ is satisfiable if and only if
$H$ is a $(0,3)$-well-covered graph.

Suppose that
$I$ is satisfiable. Then, since $H$
is a $(0,3)$-graph, every
maximal independent set of $H$ has size
$3$, $2$, or $1$. If there is a maximal
independent set~$I$ in $H$ with size $1$ or $2$,
then $I$ is a maximal clique of $G'$
of size $1$ or $2$. This contradicts the
construction of $G'$, since every maximal
clique of $G'$ is a triangle.
Therefore, $G$ is well-covered.

Suppose that $H$ is $(0,3)$-well-covered.
Then $G'$ is 3-colorable, so
$G$ is also 3-colorable.
Thus, by~\cite{Stockmeyer:1973},
$I$ is satisfiable.
\end{proof}

\medskip

We next prove that
$(3,0)${\sc {wc-g}}
is {\sf {NP}}-hard.
For this, we again use the
proof of Stockmeyer~\cite{Stockmeyer:1973}, together with the following theorem.

\begin{sloppypar}
\begin{proposition}[Topp and Volkmann~\cite{Topp1990}]\label{pendant}
Let $G=(V,E)$ be an $n$-vertex graph,
$V=\{v_1,v_2,v_3, \ldots ,v_n\}$, and let
 $H$ be obtained from $G$
such that $V(H)=V\cup \{u_1,u_2,u_3, \ldots ,u_n\}$
and $E(H)=E\cup \{v_iu_i: i\in \{1,2,3,\ldots , n\}\}$.
Then $H$ is a well-covered graph where
every maximal independent set has size $n$.
\end{proposition}
\end{sloppypar}

\begin{proof}
Observe that every maximal independent set $I$ of $H$
has a subset $I_G= I \cap V$. Let ${\cal{U}}\subseteq \{1,2,3,\ldots ,n\}$ be the
set of indices $i$ such that $v_i\in I$. Since~$I$ is maximal,
the set $\{u_i: i\in \{1,2,3,\ldots ,n\}\setminus \cal{U}\}$ must
be contained in $I$, so $|I|=n$.
\end{proof}

\begin{lemma}\label{thm:3,0}
$(3,0)${\sc {wc-g}} is {\sf {NP}}-hard.
\end{lemma}
\begin{proof}
Let $I=(U,C)$ be a {\sc {3-sat}} instance;
let $G=(V,E)$ be the graph obtained from~$I$
in Stockmeyer's~\cite{Stockmeyer:1973}
\NP-completeness proof
for {\sc {3-coloring}};
and let $H$ be the graph obtained from $G$ by the
transformation
described in Proposition~\ref{pendant}.
We prove that $I$ is satisfiable if and only if
$H$ is a $(3,0)$-well-covered graph.
Suppose that $I$ is satisfiable.
Then by~\cite{Stockmeyer:1973} we have that
$G$ is 3-colorable. Since a
vertex $v \in V(H)\setminus V(G)$ has just one neighbor,
there are 2 colors left for $v$
to extend a 3-coloring of $G$, and so
$H$ is a $(3,0)$-graph. Hence,
by Proposition~\ref{pendant},
$H$ is a $(3,0)$-well-covered graph.
Suppose that~$H$ is a $(3,0)$-well-covered graph.
Then we have that
$G$ is a $(3,0)$-graph.
By~\cite{Stockmeyer:1973}, $I$ is satisfiable.
\end{proof}

\medskip

\begin{sloppypar}
Note that Theorem~\ref{mono} combined with Lemma~\ref{thm:0,3} does not imply that {\sc {$(1,3)$wc-g}}
is NP-complete.
\end{sloppypar}

\begin{lemma}\label{teo:1-3}
{\sc {$(1,3)$wc-g}} is {\sf {NP}}-complete.
\end{lemma}

\begin{proof}
As by Theorem~\ref{0elland1ell} the {\sc Well-Covered Graph} problem can be solved in polynomial time on $(1,3)$-graphs, by Fact~\ref{fact_np} {\sc $(1,3)$wc-g} is in $\NP$.

Let $I=(U,C)$ be a {\sc {3-sat}} instance. Without loss of generality, $I$ has more than two clauses.
We produce a {\sc {$(1,3)$wc-g}} instance $H$
 polynomial in the size of~$I$,
 such that $I$ is satisfiable if and only if
$H$ is $(1,3)$-well-covered.

Let $G=(V,E)$ be the graph of Stockmeyer~\cite{Stockmeyer:1973}
obtained from $I$ (see Figure~\ref{stockfig}(a)),
and let $H$ be the graph obtained from $\overline{G}$
(the complement of the graph $G$)
by adding one pendant vertex $p_v$ for each vertex $v$ of $\overline{G}$.
Note that $V(H)=V(G)\cup \{p_v: v\in V(G)\}$,
$E(H)=E(\overline{G})\cup \{p_vv: v\in V(G)\}$, and $N_{H}(p_v)=\{v\}$.

First suppose that $I$ is satisfiable.
Then by~\cite{Stockmeyer:1973},
$G$ is a $(3,0)$-graph, and~$\overline{G}$ is a $(0,3)$-graph with partition into cliques
$V(\overline{G})=(K_{\overline{G}}^1,K_{\overline{G}}^2,K_{\overline{G}}^3)$. Thus it follows that $(S=\{p_v: v\in V(G)\},K_{\overline{G}}^1,K_{\overline{G}}^2,K_{\overline{G}}^3)$ is a $(1,3)$-partition of $V(H)$. In addition, from Proposition~\ref{pendant} and by the construction of $H$, $H$ is a well-covered graph. Hence $H$ is $(1,3)$-well-covered.

Conversely, suppose that $H$ is $(1,3)$-well-covered, and let $V(H)=(S,K^1,\allowbreak K^2,\allowbreak K^3)$ be a $(1,3)$-partition for $H$.
Then we claim that no vertex $p_v\in V(H)\setminus V(G)$ belongs  to $K^i, i\in \{1,2,3\}$.
Indeed, suppose for contradiction that $p_v\in K^i$ for some $i\in \{1,2,3\}$. Then, $K^i \subseteq \{p_v, v\}$. Hence, $H\setminus K^i$ is a $(1,2)$-graph and $G\setminus \{v\}$ is an induced subgraph of a $(2,1)$-graph. But by construction of $G$, $G\setminus \{v\}$ (for any $v\in V(G)$) contains at least one $2K_3$ (that is, two vertex-disjoint copies of~$K_3$) as an induced subgraph, which is a contradiction given that $2K_3$ is clearly a forbidden subgraph for $(2,1)$-graphs. Therefore, $\{p_v: v\in V(G)\} \subseteq S$, and since $\{p_v: v\in V(G)\}$ is a dominating set of $H$, $S=\{p_v: v\in V(G)\}$.
Thus, $\overline{G}$ is a $(0,3)$-graph with partition $V(\overline{G})=(K^1,K^2,K^3)$, and therefore $G$ is a $(3,0)$-graph, i.e., a 3-colorable graph. Therefore, by~\cite{Stockmeyer:1973}, $I$ is satisfiable.
\end{proof}

\begin{corollary}\label{cor3l0}
If $r\geq 3$ and $\ell=0$, then $(r,\ell)${\sc {wc-g}} is {\sf {NP}}-hard.
If $r\in\{0,1\}$ and $\ell \geq3$, then $(r,\ell)${\sc {wc-g}} is {\sf {NP}}-complete.
\end{corollary}
\begin{proof}
$(r,\ell)${\sc {wc-g}} is {\sf {NP}}-hard in all of these cases by combining Theorem~\ref{mono}, and Lemmas~\ref{thm:0,3}, \ref{thm:3,0} and~\ref{teo:1-3}.
For $r\in\{0,1\}$ and $\ell \geq3$, the {\sc Well-Covered Graph} problem can be solved in polynomial time on $(r,\ell)$-graphs, so by Fact~\ref{fact_np} {\sc $(r,\ell)$wc-g} is in~$\NP$.
\end{proof}

Below we summarize the cases for which we have shown that {\sc {wc-}}$(r,\ell)${\sc {g}} or $(r,\ell)${\sc {wc-g}} is computationally hard.

\begin{theorem}\label{teonpc} The following classification holds:
\hspace{.1cm}
\vspace{.1cm}
\begin{enumerate}
\item
{\sc {wc-}}$(r,\ell)${\sc {g}}
with $r\geq 2$
and
$\ell\geq 1$
are \coNPbis-complete;
\item
$(0,\ell)${\sc {wc-g}} and $(1,\ell)${\sc {wc-g}}
with $\ell\geq 3$
are \NP-complete;
\item
$(2,1)${\sc {wc-g}} and $(2,2)${\sc {wc-g}} are
\coNPbis-complete;
\item
$(r,\ell)${\sc {wc-g}} with $r\geq 0$ and $\ell\geq 3$
is \NP-hard;
\item
$(r,\ell)${\sc {wc-g}} with $r\geq 3$ and $\ell\geq 0$
is \NP-hard;
\item
$(r,\ell)${\sc {wc-g}} with $r\geq 2$ and $\ell\geq 1$
is \coNPbis-hard.
\end{enumerate}
\end{theorem}
\begin{proof}
Statement~1 follows from Proposition~\ref{21conpc} and Theorem~\ref{mono}(i).
Statement~2 follows from Corollary~\ref{cor3l0}.
Statement~3 follows from Statement~1, Facts~\ref{fact2} and~\ref{fact_conp} and the fact that recognizing $(r,\ell)$-graphs is in $\P$ if $\max\{r, \ell\} \leq 2$~\cite{Brandstadt96}.
Statement~4 follows from Statement~2 and Theorem~\ref{mono}(ii)-(iii).
Statement~5 follows from Lemma~\ref{thm:3,0} and Theorem~\ref{mono}(ii)-(iii).
Finally, Statement~6 follows from Corollary~\ref{co21} and Theorem~\ref{mono}(ii)-(iii).
\end{proof}

\newpage

\section{Parameterized complexity of the problems}
\label{sec:FPT-results}

In this section we focus on the parameterized complexity of the \textsc{Well-Covered Graph} problem, with special emphasis on the case where the input graph is an $(r,\ell)$-graph. Recall that the results presented in Section 2 show that \textsc{wc-$(r,\ell)$g} is {\sf {para-coNP}}-complete when parameterized by $r$ and $\ell$. Thus, additional pa\-ra\-me\-ters should be considered. Henceforth we let $\alpha$ (resp. $\omega$) denote the size of a maximum independent set (resp. maximum clique) in the input graph $G$ for the problem under consideration. Note that \textsc{wc}-$(r,\ell)$\textsc{g} parameterized by $r, \ell$, and $\omega$ generalizes \textsc{wc-$(r,0)$g}, whose complexity was left open in the previous sections. Therefore, we focus on the complexity of \textsc{wc-$(r,\ell)$g} parameterized by $r, \ell$, and $\alpha$, and on the complexity of the natural parameterized version of \textsc{Well-Covered Graph}, defined as follows:

\begin{flushleft}
\fbox{
\begin{minipage}{\textwidth}
\textsc{$k$-Well-Covered Graph}\vspace{.1cm}\\
\begin{tabular}{rl}
{\bf Input:}    & A graph $G$ and an integer $k$.\\
{\bf Parameter:} & $k$.\\
{\bf Question:} & Does every maximal independent set of $G$ have size exactly $k$?
\end{tabular}
\end{minipage}}
\medskip
\end{flushleft}

The next lemma provides further motivation to study of the \textsc{wc-$(0,\ell)$g} problem, as it shows that \textsc{$k$-Well-Covered Graph} (on general graphs) can be reduced to the \textsc{wc-$(0,\ell)$g} problem parameterized by $\ell$.

\begin{lemma}\label{lem:reductiongeneral}
The \textsc{$k$-Well-Covered Graph} problem can be fpt-reduced to the \textsc{wc-$(0,\ell)$g} problem parameterized by~$\ell$.
\end{lemma}
\begin{proof}
Consider an arbitrary input graph~$G$ with vertices $u_1,\ldots,u_n$.
First, we find an arbitrary maximal (with respect to set-inclusion) independent set~$I$ in~$G$.
Without loss of generality we may assume that $|I|=k$ and $I=\{u_1,\ldots,u_{k}\}$. Let $\ell=k+1$.

We construct a $(0,\ell)$-graph~$G'$ with vertex set $\{v_{i,j}:i \in \{1,\ldots,\ell\},\allowbreak j \in \{1,\ldots,n\}\}$ as follows:
\begin{itemize}
\item For all $i \in \{1,\ldots,\ell\}$ add edges to make $V_i:=\{v_{i,j}:j \in \{1,\ldots,n\}\}$ into a clique.
\item For all $j \in \{1,\ldots,n\}$ add edges to make $W_j:=\{v_{i,j}:i \in \{1,\ldots,\ell\}\}$ into a clique.
\item For all pairs of adjacent vertices $u_a$, $u_b$ in~$G$, add edges between~$v_{i,a}$ and~$v_{j,b}$ for all $i,j \in \{1,\ldots,\ell\}$ (so that~$V_a$ is complete to~$V_b$).
\end{itemize}
Note that the sets~$V_i$ partition~$G'$ into~$\ell$ cliques, so~$G'$ is indeed a $(0,\ell)$-graph, where $\ell=k+1$.

The graph $G'$ has a maximal independent set of size~$k$, namely $\{v_{1,1},\ldots,\allowbreak v_{k,k}\}$, so~$G'$ is well-covered if and only if every maximal independent set in~$G'$ has size exactly~$k$.
Every maximal independent set in~$G'$ has at most one vertex in any set~$V_i$ and at most one vertex in any set~$W_j$, since~$V_i$ and~$W_j$ are cliques.
As there are~$\ell=k+1$ sets~$V_i$, it follows that every independent set in~$G'$ contains at most~$k+1$ vertices.
If~$G'$ contains an independent set $\{v_{i_1,j_1},\ldots,v_{i_x,j_x}\}$ for some~$x$ then $\{u_{j_1},\ldots,u_{j_x}\}$ is an independent set in~$G$.
If~$G$ contains an independent set $\{u_{j_1},\ldots,u_{j_x}\}$ for some~$x$ then $\{v_{1,j_1},\ldots,v_{\min(x,k+1),j_{\min(x,k+1)}}\}$ is an independent set in~$G'$.
Therefore~$G$ contains a maximal independent set smaller than~$k$ if and only if~$G'$ contains a maximal independent set smaller than~$k$
and~$G$ contains a (not necessarily maximal) independent set of size at least~$k+1$ if and only~$G'$ contains a maximal independent set of size exactly~$k+1$.
It follows that~$G'$ is well-covered if and only if~$G$ is. As $\ell=k+1$, this completes the proof.
\end{proof}

Recall that the {\sc Well-Covered Graph} problem is {\sf coNP}-complete~\cite{chvatal1993note,SaSt92}.
In order to analyze the parameterized complexity of the problem, we will need the following definition.

\begin{definition}
The class {\sf coW[2]} is the class of all parameterized problems whose complement is in {\sf W[2]}.
\end{definition}

For an overview of parameterized complexity classes, see~\cite{de2015machine,FG06}.

We are now ready to show the next result.

\begin{theorem}\label{coW[2]}
The \textsc{wc-$(0,\ell)$g} problem parameterized by~$\ell$ is {\sf {coW[2]}}-hard.
\end{theorem}
\begin{proof}

\textsc{Red-Blue Dominating Set (RBDS)}  is a well-known {\sf W[2]}-complete problem~\cite{DF13}, which consists of determining whether a given bipartite graph $G = (R \cup B, E)$
admits a set $D\subseteq R$ of size $k$ (the parameter) such that $D$ dominates $B$ (that is, every vertex in $B$ has a neighbor in $D$).
To show the {\sf {coW[2]}}-hardness of our problem, we present an fpt-reduction from \textsc{Red-Blue Dominating Set} to the problem of determining whether a given $(0,\ell)$-graph is \emph{not} well-covered, where $\ell=k+1$.

\begin{sloppypar}
From an instance $(G,k)$ of {\sc RBDS} we construct a $(0,\ell)$-graph $G'$ as follows.
Replace the set $R=\{r_1,r_2,\ldots,r_m\}$ by $k$ copies: $R_1 =\{r_1^1,r_2^1,\ldots,r_m^1\}$, $R_2=\{r_1^2,r_2^2,\ldots,r_m^2\}, \ldots , R_k=\{r_1^k,r_2^k,\ldots,r_m^k\}$, where each new vertex has the same neighborhood as the corresponding vertex did in $G$.
Add edges to make $B$, as well as each $R_i$ for $1\leq i \leq k$, induce a clique.
For each $i \in \{1,\ldots,k\}$, create a vertex $s_i$, and add all possible edges between~$s_i$ and the vertices in $R_i$.
Let $G'$ be the resulting graph.
Note that the vertex set of $G'$ can be partitioned into $\ell=k+1$ cliques: $B, R_1 \cup \{s_1\}, R_2 \cup \{s_2\}, \ldots, R_k \cup \{s_k\}$.
\end{sloppypar}

Clearly, for every $b\in B$, the set $\{s_1, s_2,\ldots,s_k\} \cup \{b\}$ is an independent set of $G'$ of size $k+1$.
Note that such an independent set is maximum, as it contains one vertex from each of the $k+1$ cliques that partition $V(G')$. In addition, any maximal independent set of $G'$ has size at least $k$, since every maximal independent set contains either $s_i$ or a vertex of $R_i$.
At this point, we claim that $G$ has a set $D\subseteq R$ of size $k$ which dominates~$B$ if and only if $G'$ has a maximal independent set of size $k$ (i.e., $G'$ is not well-covered).

If $D = \{r_{i_1}, r_{i_2}, \ldots, r_{i_k}\}$ is a subset of $R$ of size $k$ which dominates $B$ in $G$, then $D'=\{r_{i_1}^1, r_{i_2}^2, \ldots, r_{i_k}^k\}$ is a maximal independent set of $G'$, implying that $G'$ is not well-covered.

Conversely, if $G'$ is not well-covered then there exists in $G'$ a maximal independent set $D'$ of size $k$. Note that $D'\cap B = \emptyset$ and each vertex in $B$ has at least one neighbor in $D'$, as otherwise $D'$ would not be a maximal independent set of size $k$.
Therefore, by letting $D$ be the set of vertices in $R$ that have copies in $D'\cap \{R_1 \cup R_2 \cup \ldots \cup R_k\}$, we find that $D$ is a subset of $R$ of size at most $k$ which dominates $B$ in $G$.
\end{proof}

From the previous theorem we immediately obtain the following corollaries.

\begin{corollary}\label{k-wcW[2]}
The \textsc{$k$-Well-Covered Graph} problem is {\sf {coW[2]}}-hard.
\end{corollary}
\begin{proof}
This follows immediately from Lemma~\ref{lem:reductiongeneral} and Theorem~\ref{coW[2]}.
\end{proof}

\begin{corollary}\label{noAlgo}
Unless {\sf FPT} = {\sf coW[2]}, the \textsc{wc-$(r,\ell)$g} problem cannot be solved in time $f(\alpha + \ell) n^{g(r)}$ for any computable function $f$.
\end{corollary}
\begin{proof}
This follows from the fact that an algorithm running in time $f(\alpha + \ell) n^{g(r)}$, would be an {\sf FPT}-algorithm for \textsc{wc-$(0,\ell)$g} parameterized by~$\ell$, and from the {\sf coW[2]}-hardness of the problem demonstrated in Theorem~\ref{coW[2]}.
\end{proof}

In contrast to Corollary~\ref{noAlgo}, Lemma~\ref{lem:FPT-size} shows that the \textsc{wc-$(r,\ell)$g} problem can be solved in time $2^{r \alpha} n^{O(\ell)}$.

\begin{lemma}\label{lem:FPT-size}
The \textsc{wc-$(r,\ell)$g} problem can be solved in time $2^{r \alpha} n^{O(\ell)}$. In particular, it is {\sf FPT} when $\ell$ is fixed and $r, \alpha$ are parameters.
\end{lemma}

\begin{proof}
Note that each of the $r$ independent sets $S^1, \ldots, S^r$ of the given partition of $V(G)$ must have size at most $\alpha$. On the other hand, any maximal independent set of $G$ contains at most one vertex in each of the $\ell$ cliques. The algorithm exhaustively constructs all maximal independent sets of $G$ as follows: we start by guessing a subset of $\bigcup_{i=1}^r S^i$, and then choose at most one vertex in each clique.
For each choice,
we just have to verify whether
the constructed set is a maximal independent set,
and then check that all the constructed maximal independent sets have the same size.
The claimed running time follows. In fact, in the statement of the lemma, one could replace $r\alpha$ with $\sum_{1 \leq i \leq r}|S^i|$, which yields a stronger result.
\end{proof}

Although \textsc{wc-$(1,\ell)$g} parameterized by $\ell$ is {\sf coW[2]}-hard (see Theorem~\ref{coW[2]}), Theorem~\ref{0elland1ell} shows that the problem is in \XP.

\begin{corollary}\label{lem:algo-XP}
The \textsc{wc-$(1,\ell)$g} problem can be solved in time $ n^{O(\ell)}$. In other words, it is in {\sf XP} when parameterized by $\ell$.
\end{corollary}
\begin{proof}
This follows from Theorem~\ref{0elland1ell} by considering $\ell$ to not be a constant.
\end{proof}

Table~\ref{table_param} summarizes the results presented so far. Note that, by Ramsey's Theorem~\cite{Ra30}, when both $\omega$ and $\alpha$ are parameters the input graph itself is a trivial kernel.

\begin{table}[htb]
\centering
\caption{Parameterized complexity of {\sc wc-($r$,$\ell$)g}.}
\label{table_param}
\begin{tabular}{|c|c|c|c|}
\hline
Param.\textbackslash Class & $(0,\ell)$                                                  & $(1,\ell)$                                                  & $(r,\ell)$                                                                                                                       \\ \hline
$r$                            & --                                                          & --                                                          & {\sf para-coNP}-h                                                                                                                      \\ \hline
$\ell$                         & \begin{tabular}[c]{@{}c@{}}{\sf coW[2]}-h\\ {\sf XP}\end{tabular}         & \begin{tabular}[c]{@{}c@{}}{\sf coW[2]}-h\\ {\sf XP}\end{tabular}         & {\sf para-coNP}-h                                                                                                                      \\ \hline
$r,\ell$                       & \begin{tabular}[c]{@{}c@{}}{\sf coW[2]}-h\\ {\sf XP}\end{tabular}         & \begin{tabular}[c]{@{}c@{}}{\sf coW[2]}-h\\ {\sf XP}\end{tabular}         & {\sf para-coNP}-h                                                                                                                      \\ \hline
$r,\ell,\omega$                & \begin{tabular}[c]{@{}c@{}}{\sf FPT}\\ Trivial\end{tabular}     & \begin{tabular}[c]{@{}c@{}}{\sf FPT}\\ Trivial\end{tabular}     & \begin{tabular}[c]{@{}c@{}}Open\\ (generalizes {\sc wc-(3,0)g})\end{tabular}                                                             \\ \hline
$r,\ell,\alpha$                & \begin{tabular}[c]{@{}c@{}}{\sf coW[2]}-h\\ {\sf XP}\end{tabular}         & \begin{tabular}[c]{@{}c@{}}{\sf coW[2]}-h\\ {\sf XP}\end{tabular}         & \begin{tabular}[c]{@{}c@{}}{\sf coW[2]}-h, \vspace{0.2cm}\\ no $f(\alpha
+ \ell)n^{g(r)}$ algo.\\ unless {\sf FPT}={\sf coW[2]},\vspace{0.2cm} \\ algo. in time $2^{r\alpha}n^{O(\ell)}$\end{tabular} \\ \hline
$\omega,\alpha$                & \begin{tabular}[c]{@{}c@{}}{\sf FPT}\\ Ramsey's Thm.\end{tabular} & \begin{tabular}[c]{@{}c@{}}{\sf FPT}\\ Ramsey's Thm.\end{tabular} & \begin{tabular}[c]{@{}c@{}}{\sf FPT}\\ Ramsey's Thm.\end{tabular}                                                                      \\ \hline
\end{tabular}
\end{table}

\subsection{Taking the neighborhood diversity as the parameter}
\label{sec:neigh-diver}

Neighborhood diversity is a structural parameter based on a special way of partitioning a graph into independent sets and cliques. Therefore, it seems a natural parameter to consider for our problem, since an $(r, \ell)$-partition of a graph $G$ is also a partition of its vertex set into cliques and independent sets.

\begin{definition}[Lampis~\cite{Lampis12}]
The \emph{neighborhood diversity} $\nd(G)$ of a graph $G=(V,E)$ is the minimum integer~$t$ such that
$V$  can be partitioned into $t$ sets $V_1,\ldots,V_t$ where for every $v\in V(G)$ and every $i\in \{1,\ldots,t\}$, either $v$ is adjacent to every vertex in $V_i$ or it is adjacent to none of them. Note that each part $V_i$ of $G$ is either a clique or an independent set.
\end{definition}

Another natural parameter to consider is the vertex cover number, because well-covered graphs can be equivalently defined as graphs in which every minimal vertex cover has the same size.
However, neighborhood diversity is {\sl stronger} than vertex cover, in the sense that every class of graphs with bounded vertex cover number is also a class of graphs with bounded neighborhood diversity, but the reverse is not true~\cite{Lampis12}.
Thus, for our analysis, it is enough to consider the neighborhood diversity as the parameter. In addition,
neighborhood diversity is a graph parameter that captures more precisely than  vertex cover number the property that two vertices with the same neighborhood are ``equivalent''.

It is worth mentioning that an optimal neighborhood diversity decomposition of a graph $G$ can be computed in time $O(n^3)$; see~\cite{Lampis12} for more details.

\begin{lemma}\label{prop:FPT-dual-size}
The \textsc{Well-Covered Graph} problem is {\sf FPT} when parameterized by neighborhood diversity.
\end{lemma}

\begin{proof}
Given a graph $G$, we first obtain a neighborhood partition of $G$ with minimum width using the polynomial-time algorithm of Lampis~\cite{Lampis12}.  Let $t :=\nobreak \nd(G)$ and let $V_1,\ldots,V_t$ be the partition of $V(G)$. As we can observe, for any pair $u,v$ of non-adjacent vertices belonging to the same part $V_i$, if $u$ is in a maximal independent set $S$ then $v$ also belongs to $S$, otherwise $S$ cannot be maximum. On the other hand, if $N[u] = N[v]$ then for any maximal independent set $S_u$ such that $u\in S_u$ there exists another maximal independent set $S_v$ such that $S_v = S_u\setminus \{u\} \cup \{v\}$. Hence, we can contract each partition~$V_i$ that is an independent set into a single vertex $v_i$ with weight $\tau(v_i)=|S_i|$, and contract each partition $V_i$ that is a clique into a single vertex $v_i$ with weight $\tau(v_i)=1$, in order to obtain a graph~$G_t$ with $|V(G_t)|=t$, where the weight of a vertex $v_i$ of~$G_t$ means that any maximal independent set of $G$ uses either none or exactly~$\tau(v)$ vertices of $V_i$. At this point, we just need to analyze whether all maximal independent sets of $G_t$ have the same weight (sum of the weights of its vertices), which can be done in time $2^t n^{O(1)}$.
\end{proof}

\begin{corollary}
The \textsc{Well-Covered Graph} problem is {\sf FPT} when parameterized by the vertex cover number $n-\alpha$.
\end{corollary}

\subsection{Taking the clique-width as the parameter}
\label{sec:clique-width}

In the 90's, Courcelle proved that for every graph property~$\Pi$ that can be formulated in {\em  monadic second order logic}  ($\textsc{MSOL}_1$), there is an $f(k)n^{O(1)}$ algorithm that decides if a graph~$G$ of clique-width at most~$k$ satisfies~$\Pi$ (see~\cite{C90,C97,Co11,C93}), provided that a $k$-expression is given.

$\textsc{LinEMSOL}$ is an extension of $\textsc{MSOL}_1$ which allows searching for sets of vertices which are optimal with respect to some linear evaluation functions.
Courcelle et al.~\cite{courcelle2000linear} showed that every graph problem definable in $\textsc{LinEMSOL}$ is linear-time solvable on graphs with clique-width at most~$k$ (i.e., {\sf FPT} when parameterized by clique-width) if a $k$-expression is given as input. Using a result of Oum~\cite{Oum08}, the same result follows even if no $k$-expression is given.

\begin{theorem}\label{the:cwd}
The \textsc{Well-Covered Graph} problem is {\sf FPT} when parameterized by clique-width.
\end{theorem}
\begin{proof}
	Given $S\subseteq V(G)$, first observe that the property ``$S$ is a maximal independent set'' is $\textsc{MSOL}_1$-expressible.
	Indeed, we can construct a formula~$\varphi(G,S)$ such that~``$S$ is a maximal independent set'' $\Leftrightarrow \varphi(G,S)$ as follows:	
	\begin{equation*}\label{MSOL1}
		\begin{split}
			[\nexists~ u,v \in S : edge(u,v)]~\wedge~
			[\nexists~ S': (S\subseteq S')~\wedge (\nexists~ x,y \in S' : edge(x,y))]
		\end{split}
	\end{equation*}
	
Since $\varphi(G,S)$ is an $\textsc{MSOL}_1$-expression, the problem of finding $goal(S): \varphi(G,S)$ for $goal\in\{\max, \min\}$ is definable in $\textsc{LinEMSOL}$. Thus we can find $\max(S)$ and $\min(S)$ satisfying $\varphi(G,S)$ in time $f(\cw(G)) n^{O(1)}$. Finally, $G$ is well-covered if and only if $|\max(S)|=|\min(S)|$.
\end{proof}

\begin{corollary}
The \textsc{Well-Covered Graph} problem is {\sf FPT} when parame\-te\-ri\-zed by treewidth.
\end{corollary}
\begin{proof}
This follows from the fact that graphs with treewidth bounded by $k$ have clique-width bounded by a function of $k$~\cite{CO00-2}.
\end{proof}

\begin{corollary}
For any fixed~$r$ and~$\ell$, the \textsc{($r$,$\ell$)-Well-Covered Graph} problem is {\sf FPT} when parameterized by clique-width.
\end{corollary}
\begin{proof}
As $r$ and $\ell$ are constants, the problem of determining whether $G$ is an $(r,\ell)$-graph is also $\textsc{MSOL}_1$-expressible.
\end{proof}

Note that, since for every graph $G$ we have $\cw(G)\leq \nd(G)+1$~\cite{Lampis12}, Lemma~\ref{prop:FPT-dual-size}  is also a corollary of Theorem~\ref{the:cwd}. Nevertheless, the algorithm derived from the proof of Lemma~\ref{prop:FPT-dual-size} is much simpler and faster than the one that follows from the meta-theorem of Courcelle et al.~\cite{courcelle2000linear}.

\section{Further research}
\label{sec:further}

Concerning the complexity of the {\sc {$(r,\ell)$wc-g}} and {\sc wc-{$(r,\ell)$g}} problems, note that the only remaining open cases are
{\sc {wc-$(r,0)$g}} for $r\geq 3$ (see the tables in Section~\ref{sec:intro}).  We do not even know if there exists some integer $r\geq 3$ such that {\sc {wc-$(r,0)$g}} is {\sf {coNP}}-complete, although we conjecture that this is indeed the case.

As another avenue for further research, it would be interesting to provide a complete characterization of well-covered tripartite graphs, as has been done for bipartite graphs~\cite{Fav82,Rav77,Vil07}. So far, only  partial characterizations exist~\cite{Hag14,GoZa15}.

\bigskip

\noindent {\small\textbf{Acknowledgement}. We would like to thank the anonymous reviewers for helpful remarks that improved the presentation of the manuscript.}
\bibliographystyle{abbrv}	
\bibliography{bib-FPT}

\end{document}